\documentclass[a4paper, USenglish]{lipicsoid-2019}

\hideLIPIcs
\usepackage{etex}
\usepackage{amsmath}
\usepackage{amssymb,amsthm,url}
\usepackage{stmaryrd}
\usepackage{todonotes}
\usepackage{comment}
\usepackage[numbers, sort&compress]{natbib} 
\usepackage{doi} 
\usepackage{bibentry}
\nobibliography*

\hypersetup{
  colorlinks=true,
  linkcolor=black,
  citecolor=black,
  filecolor=black,
  urlcolor=[rgb]{0,0.1,0.5},
  pdftitle={Silent MST approximation for tiny memory},
  pdfauthor={}
}

\newcommand{\ID}{ID}

\title{Silent MST approximation for tiny memory}

\author{Lélia Blin}%
{LIP6, Sorbonnes University}%
{lelia.blin@lip6.fr}%
{http://orcid.org/0000-0003-0342-9243}%
{}

\author{Swan Dubois}%
{LIP6, Sorbonnes University}%
{swan.dubois@lip6.fr}%
{}%
{}

\author{Laurent Feuilloley}%
{DII, Universidad de Chile}%
{feuilloley@dii.uchile.fr}%
{http://orcid.org/0000-0002-3994-0898}%
{}

\authorrunning{L. Blin, S. Dubois and L. Feuilloley}
\funding{Support by ANR ESTATE}



\begin{document}
\maketitle


\begin{abstract}
In this paper we show that approximation can help reduce the space used for self-stabilization. In the classic \emph{state model}, where the nodes of a network communicate by reading the states of their neighbors, an important measure of efficiency is the space: the number of bits used at each node to encode the state. In this model, a classic requirement is that the algorithm has to be \emph{silent}, that is, after stabilization the states should not change anymore. We design a silent self-stabilizing algorithm for the problem of minimum spanning tree, that has a trade-off between the quality of the solution and the space needed to compute it. 
\end{abstract}


\setcounter{page}{1}
\section{Introduction}

\subsection{Our questions}
\paragraph*{Context.}

Self-stabilization is a technique to ensure fault-tolerance in distributed systems. 
It aims at designing systems that can recover from arbitrary faults. 
Silent self-stabilization consists in asking for the additional property that, once a correct configuration has been reached, the processors basically stop computing.  

In the context of self-stabilization, the most studied measure of performance is the time to stabilize to a correct configuration. 
Another essential parameter is the space used by each processor. 
This parameter not only captures some notion of memory (and is actually also called \emph{the memory}), but more remarkably, it captures the performance in terms of communication, as self-stabilizing algorithms communicate by reading the states of their neighbors (when they are described in the so-called \emph{state model}, which is the most common model). 

For silent self-stabilizing algorithms, this memory usage is tightly related to the space needed to locally certify that a configuration is correct. 
Such certifications, also called proofs, have been studied independently under the name of \emph{proof-labeling schemes}. 
On the one hand, it is known that the space needed for the proof is a lower bound on the space required for silent stabilization. Indeed, after stabilization, a silent algorithm is only checking that the configuration is correct via reading its neighbors' states, which is exactly what a distributed proof is made for.
On the other hand, it is proved in~\cite{BlinFP14} that one can always design an algorithm matching this lower bound (up to an additive logarithmic factor), even in the most asynchronous setting. 
Thus in some sense, one can always achieve optimal space. 

There are two issues to this situation. First  the general technique of \cite{BlinFP14} is inherently exponential in time: it basically consists in looking for the distributed proof via an exhaustive search. Second, the space required for silent stabilization can be simply be too large for applications.

\paragraph*{Approximation-memory trade-off} 

The core of our paper is to give a solution for the second problem, which can be rephrased as: what can be done when we do not even have the space needed for a distributed proof? 
One technique is to consider non-silent algorithms, that keep changing their states. For example, \cite{BlinT18} achieves $O(\log \log n)$ space for leader election on a ring, when the lower bound for silent stabilization is $\Omega(\log n)$ (where $n$ is the number of nodes in the network).
In this paper, we make the choice of keeping the silence property, but to be less demanding on the quality of the solution.
More precisely we are aiming at a trade-off between the memory used and the quality of the solution produced, that is, we want to design approximation algorithms for optimization problems, such that the larger the memory allowed, the better the approximation ratio.
To our knowledge this is the first time approximation is used to reduce memory usage for self-stabilization (although it has recently been proved fruitful in the more restricted context of proof-labeling schemes~\cite{Censor-HillelPP20,  EmekG20}).

\paragraph*{Optimal space in polynomial time}

Now a second question, which follows from the exponential-time algorithm of \cite{BlinFP14}, is: when we can afford the optimal space to compute an exact solution, can we get it in  polynomial time?
The answer is no in general. Consider for example the task of 3-coloring a 3-colorable graph. 
The distributed proof uses only constant space, because the colors are enough for local checkability.
On the other hand, it is known that no algorithm can compute a 3-coloring in constant space. Indeed,  in order to perform even a minimal symmetry breaking (such as having two nodes with two different outputs), an algorithm needs strictly more than constant space \cite{BeauquierGJ99} (actually $\Omega(\log \log n)$ bits are necessary \cite{BlinFB19}). 
On the positive side, \cite{BlinF15} shows that for various tree construction problems, one can match the optimal space bound and have polynomial-time stabilization. 
In particular, one can get down to $\Theta(\log^2\!n)$ bits for minimum spanning tree, which is optimal when the edge weights are in a polynomial range. As we will see, we can improve on this, as a side result of our approximation algorithm.

\subsection{Our results}

In this paper, we focus on the central problem of minimum spanning tree (MST).
Our main result is an approximation-memory trade-off for this problem. 
The theorem below, and all the results of this paper hold under the classic assumption that the edge weights are in $[1,n]$, where $n$ is the size of the network.\footnote{Assuming the maximum to be $n$ and not poly(n) allows to have cleaner proofs without additional constants, but the asymptotic results are also correct for polynomial weights.} 

\begin{theorem}\label{thm:approximation}
There exists a silent self-stabilizing approximation algorithm for minimum spanning tree, that stabilizes in polynomial time and has a trade-off between memory and approximation. 
This trade-off goes from space $O(\log^2\!n)$ for a minimum spanning tree to space $O(\log n)$ for a simple spanning tree.
\end{theorem}

The precise trade-off has a complicated expression, thus we do not write explicitly here. It is given in Lemma~\ref{lem:black-box}.
The two extreme values, $O(\log^2n)$ for an MST and $O(\log n)$ for a simple spanning tree are optimal (see~\cite{KormanKP10, KormanK07} for the lower bounds).
We get a smooth trade-off between these extremes, with for example $O(\log n \log \log n)$ space for a 2-approximation.

One of the two ingredients to achieve this result is an exact algorithm for MST, which is self-stabilizing, silent and polynomial-time, and uses $O(\log n \cdot s)$ space,\footnote{Here and everywhere in the paper, $\log n \cdot s$ should be read as $(\log n) \times s$.} where $s$ is the number of bits used to encode an edge weight.\footnote{Note that as we assume the weights are polynomial in $n$, $s$ is in $O(\log n)$.}

\begin{theorem}\label{thm:optimal-algorithm}
There exists a silent self-stabilizing algorithm for (exact) minimum spanning tree, with $O(\log n \cdot s)$ memory, that stabilizes in polynomial time.
\end{theorem}

It is known that an MST requires $\Omega(\log n \cdot s)$ space \cite{KormanK07}.
Therefore our algorithm improves on the state-of-the-art by proving that, even with a parametrization by~$s$, MST is part of the set of problems that can be solved in optimal space \emph{and} polynomial-time.

\subsection{Our techniques}

Our algorithm has two main ingredients: the exact algorithm that we have already mentioned and a technique to transform the weights. The weight transformation changes the original weights into approximated weights that can be encoded in smaller space. Then we basically feed these approximated weights to the exact algorithm and get as a result an approximate solution. The better the approximation of the weight, the better the approximation of the final solution, but the larger the space used.

The weight transformation (Lemma~\ref{lem:black-box}) takes as input a weight in $[1,poly(n)]$, encoded on $\Theta(\log n)$ bits, and outputs a weight in a smaller range, hence using less bits of memory. 
The simplest form of the technique is the following: replace each weight by the position of its most significant bit. 
This way when we write weights in the memory, we use exponentially less space: $s$ will be in $O(\log \log n)$ instead of $\Theta(\log n)$. 
Of course by this operation we loose some precision. 
Namely, we only have the information to recover a 2-approximation of each weight. 
Now if we feed these ``new weights'' to an exact algorithm for minimum spanning tree, then we will get a 2-approximation, and using much less space.
We design an extension of this technique, that allows to get the whole trade-off between space and approximation. This extension is more complicated, but is still based on manipulations of the binary representation of the weights.  

The rest of the paper consists in designing the exact self-stabilizing algorithm for minimum spanning tree in space $O(\log n \cdot s)$. 
This exact algorithm does not follow the usual design of silent self-stabilizing algorithms.
A silent algorithm typically stores some key pieces of information, \emph{while} it is building the output, \emph{e.g.} while selecting the edges of the MST. 
These pieces of information form a certificate of correctness that allows, during and after stabilization, to check whether the construction is correct. And if the construction is correct then the output is correct. 
This is for example the way the $O(\log^2\!n)$-space algorithm of~\cite{BlinF15} is designed, book-keeping the important information of a Boruvka-inspired algorithm.
Unfortunately, it seems that  this approach is difficult if not impossible to use when one wants to go below $\log^2\!n$ space for minimum spanning tree.
Instead, we use a two-phase approach: we first build a minimum spanning tree, and, once it is finished, we certify it. 
This paper is, as far as we know, the first occurrence of such a modular approach. 
The certification we use is the proof-labeling scheme of~\cite{KormanK07}. As a side result we answer a question of~\cite{KormanK07} where designing a self-stabilizing algorithm using this certification was left as an open problem.
  
\subsection{Outline}

We start in Section~\ref{sec:model} with a description of the model.
In Section~\ref{sec:general-description}, we describe the general structure of our algorithm. 
Then in Section~\ref{sec:approximation}, \ref{sec:construction}, and ~\ref{sec:labeling-algorithm}, we describe the different components of our algorithm. 
Section~\ref{sec:distributed-proof} is a high-level description of the certification of~\cite{KormanK07} that we use in Section~\ref{sec:labeling-algorithm}.


\section{Model}
\label{sec:model}

We consider that a network is represented by an undirected connected graph $G=(V,E)$ where $V$ is a set of processors (or nodes), and $E$ is a set of edges that represent communication channels.
We denote the number of nodes by $n$.
Every node $v$ is given a unique identifier $\ID_v$, and every edge has a weight. 
Both identifiers and weights are polynomial, that is, they are integers in $[1,poly(n)]$. The identifiers are all distinct, but the weights are not required to be distinct.

We want to compute an approximation of a minimum spanning tree.
A $k$-approximation of an MST is a spanning tree whose weight is not larger than $k$ times the weight of an MST. 
Note that in our definition of approximation, we do not relax the requirement of acyclicity, thus it is not the same kind of approximation as the one used in the literature about spanners.

\paragraph*{State model}

Our algorithm is designed for the classic \emph{state model}~\cite{Dolev00}.
In this model, every node has a state, and communication between neighbors is modeled by direct reading of states instead of exchanges of messages. 
This state is the \emph{mutable memory}, that is, it is the part of the memory that can be modified, and also the one that is counted when we consider space complexity. 
There is also the \emph{non-mutable memory}, that contains for each node, its identifier and the weights of the adjacent edges, as well as the code of the algorithm.
The tuple of all the states of the network is called the \emph{configuration}, and the execution of an algorithm is therefore described by a sequence of configurations.

In the state model, an algorithm is usually described as a set of rules. 
A rule basically states that if the local view of a node satisfies some property, then the node can change its state to a specified new value.
Here by ``local view'' we mean the state of a node, the states of its neighbors, the identifier of the node, and the weights of the adjacent edges.
If there exists a rule, such that the property of the node's view is satisfied, then we say that the node is \emph{enabled}. 
(Note that for a node deciding whether it is enabled or not could take some time and space. As in most models in network distributed computing, such local computation is not taken into account for the time and space complexity.) 
The asynchrony of the system is modeled by a scheduler who chooses, at each step, a non-empty subset of enabled nodes, that are allowed to apply a rule.
We consider the harshest scheduler, the \emph{distributed unfair scheduler}, which has no further constraint for the choices it makes. Other schedulers considered in the literature can have, for example, fairness constraints: they cannot activate always the same nodes. See~\cite{DuboisT11} for a survey of the schedulers of the literature. The choice of the distributed unfair scheduler makes our algorithm the most robust possible. 
 
To compute time complexities, we use the definition of \emph{round} of~\cite{DolevIM97}. 
This definition captures the execution rate of the slowest process in any execution. The first round of an execution $\epsilon$, noted $\epsilon'$, is the minimal prefix of $\epsilon$ such that every node that was enabled in the initial configuration, either has taken a step or is not enabled anymore. Let $\epsilon''$ be the suffix of $\epsilon$ such that $\epsilon=\epsilon'\epsilon''$. The second round of $\epsilon$ is the first round of $\epsilon''$, and so on.

\paragraph*{Distributed proof, proof-labeling schemes and silent stabilization}

Given a property, for example `a set of pointers  defines a spanning tree', a \emph{distributed proof} (or \emph{distributed certificate}) is a labeling of the nodes that certifies that the property is satisfied. 
It is usually presented as a \emph{proof-labeling scheme} \cite{KormanKP10}. 
In such a scheme, the first element is an oracle, called the \emph{prover}, which provides each node with a label. 
The second element of a scheme is a \emph{verification algorithm}. 
This algorithm, run at a node~$v$, takes as input the view of~$v$, including the labels of $v$ and of its neighbors, and decides whether to \emph{accept} or to \emph{reject}.
A scheme is correct for a property $P$, if (1) for any configuration satisfying $P$, there is a way for the prover to make all nodes accept, and (2) for any configuration \emph{not} satisfying $P$, there is \emph{no} way for the prover to make all nodes accept.
The performance of a scheme is measured by the size of the labels in number of bits (all labels have the same size).
The notion of distributed proof is tightly related to the concept of \emph{silent self-stabilizing algorithm}. 
An algorithm is \emph{self-stabilizing} if starting from an arbitrary configuration, it reaches a correct configuration after some finite time, called the \emph{stabilization time}, and stays in correct configurations afterwards. 
Such an algorithm is \emph{silent}, if the algorithm reaches a correct configuration and then stays silent: it does not change the states anymore, or in other words, no node is enabled. 
To be sure to be in a correct configuration, a silent self-stabilizing algorithm has to keep some certification of the correctness. This certification ensures that, if the configuration is not correct, then at least one node will detect it, be enabled, and start the recovery (for example launching a reset).
This is basically the same as the notion of distributed proof above~\cite{BlinFP14}, except that the proof is not given by an oracle: it is built by the algorithm.
We refer to \cite{FeuilloleyF16, Feuilloley19b} for surveys on distributed proofs.


\section{General description}
\label{sec:general-description}

Our algorithm is made of several components.
Basically, we have several algorithms that will operate one after the other, in order to reach a configuration with a minimum spanning tree certified by a distributed proof. 
The algorithms are designed to work if they start from a clean situation (for example our first algorithm expects its variables to be empty) and we have a reset mechanism that will go back to such a situation if one of our algorithms detects a problem.

\paragraph*{Two-phase approach}

In order to reach a configuration where the nodes can safely stop updating their states, we need first to have a correct solution at hand, and second to allow the nodes to be sure that indeed the solution is correct. As said earlier, the classic way to do this is to keep in memory some key extra pieces of information gathered during the computation. As it seems that this approach is difficult to implement here, we use another way: we first build the solution alone, and then we build a certification (\emph{i.e.} a distributed proof) of this solution. 

The strategy of aiming for a distributed proof simplifies the design of the algorithm.
Indeed, if something goes wrong in the computation because of the initial configuration, then we have to face two rather simple situations.
In the first case, one of our algorithms detects the error, for example because there is an obvious inconsistency between the neighbors' states.
In the second case, the problem is subtle enough to not be detected during the run of our algorithms, but then if the output is not correct, the distributed proof that is built cannot be correct either. 
In both cases the error is detected, and a reset is launched. 
In other words, either there is something obviously wrong that is caught on the fly, or there is something that is more subtle, and it is caught at the end.

The difficulties that remain are the same as for any self-stabilizing algorithm. First, one has to ensure that for any computation, starting from a clean configuration, we cannot end up in a configuration that is detected as incorrect. 
If this were to happen, then the scheduler could make the algorithm go through a reset infinitely often, and it would never stabilize. 
Second, we have to make sure that the algorithm does not get stuck in a position where no node can be activated.

\paragraph*{The components and how they work together} 

We have three main components: one algorithm that builds the minimum spanning tree (detailed in Section~\ref{sec:construction}), one algorithm that takes this tree and certifies it, thus allows the silent stabilization (detailed in Section~\ref{sec:labeling-algorithm}), and a reset algorithm that can erase everything and go back to a clean configuration.  
We describe the algorithm in a modular way to ease the reading, but in the end our result is one algorithm. 
In particular, it is important to have the three pieces working together. Note that such modular design for a self-stabilizing algorithms is not new, even for unfair schedulers: for example in the recent and celebrated coloring algorithm of~\cite{BarenboimEG18} there is a first algorithm reducing the number of colors very fast and then a second algorithm to finish the job by eliminating the last extra-colors. 

In our algorithm, the reset procedure is dominant, in the sense that if a reset is launched it will basically overrule the other procedures.
For this reset we take a solution from the literature: Devismes and Johnen~\cite{DevismesJ19} recently proposed a cooperative (that is, tolerating multiple simultaneous initiators) silent self-stabilizing reset algorithm that satisfies our constraints in terms of scheduler, stabilization time and space complexity. Then to articulate the two other pieces, we simply use flags that indicate for each node which algorithm it is running. If the algorithm are not run one after the other, then there will be a local inconsistency between these flags, and the reset will be launched. 

One last element about how the pieces work together is related to the weight transformation.
As said earlier, a key ingredient to get our approximation algorithm in small space (Theorem~\ref{thm:approximation}) is a transformation of the weights. 
This transformation consists in replacing each weight by a smaller weight. 
But, as we have limited space, we cannot store the transformed weights of all the edges adjacent to a node in its state (there can be $n-1$ edges adjacent to a node).
Instead, every time a step is taken, the node recomputes the new weights to know which rule applies.


\section{Approximation-memory trade-off}
\label{sec:approximation}

In this section, we show how we can replace the original weights with approximated weights to decrease the number of bits used to encode them, while preserving guarantees on the MST computed.
The precise trade-off between the space used for the new weights and the approximation is given in Lemma~\ref{lem:black-box}. As the expression is not very elegant, this is the only place of the paper where we write it explicitly. The trade-off for the whole algorithm can be derived from the values of this lemma: simply multiply by $O(\log n)$.

\begin{lemma}\label{lem:black-box}
There exists a transformation of the weights that allows for a trade-off between space needed to encode the new weights, and the quality of the tree one can compute from them. More precisely, for every integer $k$ in the range $[-\log\log n,\log n]$, we can get new weights with the following properties:
\begin{itemize}
\item for $k=\log n -1 $, approximation 1 and size $\log n+1$,
\item $k \in [0, \log n-2]$,, approximation $1+\frac{1}{2^k}$ and space $k+ \log(\log n -k+1))+1$,
\item for $k\in [-\log\log n,0]$ approximation $2^{2^{-k}}$ and space $ \log \left(\frac{\log n}{2^{-k}}+1\right)+1$.
\end{itemize}
\end{lemma}

Before proving the lemma, let us restate the three points of the trade-offs that we have already mentioned.

\begin{corollary}\label{coro:points}
The construction of Lemma~\ref{lem:black-box} gives in particular:
\begin{itemize}
\item Exact solution, with weights encoded on $O(\log n)$ bits.
\item 2-approximation, with weights encoded on $O(\log \log n)$ bits.
\item A trivial guarantee ($poly(n)$-approximation), with weights encoded on a constant number of bits.
\end{itemize}
\end{corollary}

The exact solution directly follows from the case $k=\log n -1$. The 2-approximation corresponds to $k=0$. The arbitrary approximation corresponds to $k=-\log \log n$.

The full proof of Lemma~\ref{lem:black-box} is given below, but we first give a sketch of the argument now. The explanation is illustrated by Figure~\ref{fig:bucketing}.

The core idea of the weight transformation is to group the weights into buckets, that will be assigned the same new weight. 
The larger the buckets, the less bits needed to encode the group name, but the larger the rounding error on each weight. 
The basic idea is to use exponential bucketing. 
For example, with exponent 2, a weight $w$, will be in a bucket $b$, if $2^{b-1} < w \leq 2^{b}$. 
This way, every weight is at most doubled, and an MST computed on these new weights is at most twice as heavy as the MST with the original weights.
The good thing is that now there are $O(\log n)$ different weights instead of $n$, which means it can be encoded  on $O(\log \log n)$ bits instead of $O(\log n)$. Now for various reasons explained in the full version \cite{BlinDF19}, we do not use this vanilla version of exponential bucketing with other bases for other approximation ratios. Instead we consider the series of the rounding values 1, 2, 4, 8, ..., $2^{\log n}$  given by the technique above, that we call \emph{milestones}, and work on it. We remove some of the milestones to get a coarser approximation or we add new ones to get a more fine-grain approximation.

\begin{figure}
\begin{center}
\scalebox{0.9}{
\input{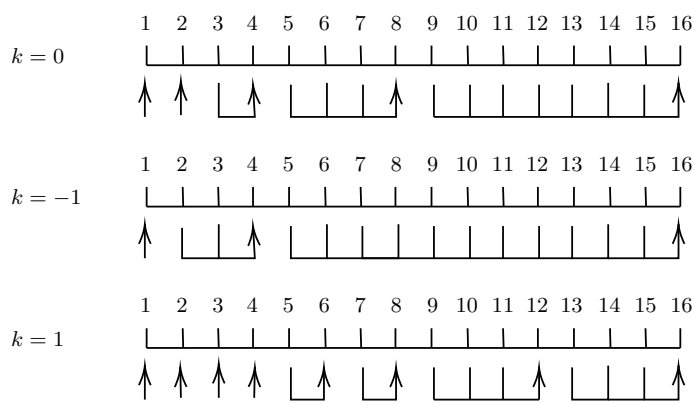}}
\caption{\label{fig:bucketing} Illustration of Lemma~\ref{lem:black-box}. For $k=0$, for example the weight 2, stays as 2, but weight 3 is transformed into 4. For the case $k-1$, we remove basically half milestones, and for $k=1$ we basically double the number of milestones. }
\end{center}
\end{figure} 

\begin{proof}[Proof of Lemma~\ref{lem:black-box}]

As said in the sketch of the proof, the technique for the weight transformation for a 2-approximation is simply to round the the weight to the next power of 2. That is, we bucket all the weights between to powers of 2 to the same weight. 

To adjust the trade-off between approximation and space, we could change the exponent: replacing 2 by a larger number, means worse approximation but better size, and smaller number means better approximation, but larger size. 
But this has two drawbacks. 
First, this cannot be used to get a trade-off that goes all the way to the exact solution. 
Indeed even taking the exponent very close to 1, one does not get all buckets of size 1. 
Second, the computations with an arbitrary exponent would use a lot of rounding, and even if in the model we do not restrict the computational power of the nodes, it is more elegant to have a transformation that works with simple manipulations of the binary representation of the weights. 
For example, the bucketing with a power of 2 is efficient, as one can decide the bucket by only considering the position of the most significant bit.
For these two reasons we use a slightly different technique.

For this technique, we consider \emph{milestones} which are simply numbers between $1$ and $n$, such that in the transformation, each weight is rounded up to the closest milestone. 
In the following claim, \emph{approximation $x$}, means that the rounding to the next milestone implies that the weight is multiplied by at most $x$.
\begin{claim}\label{clm:milestones}
Suppose $n$ is a power of 2.
For every integer $k$ in the range $[-\log\log n,\log n-1]$, there exists a set of milestones with the following properties:
\begin{itemize}
\item for $k=\log n -1 $, approximation 1 and size $n$
\item for $k \in [0, \log n-2]$, approximation $1+\frac{1}{2^k}$ and size $2^k(\log n -k+1)$.
\item for $k\in [-\log\log n,0]$, approximation $2^{2^{-k}}$ and size $\frac{\log n}{2^{-k}}+1$.
\end{itemize}
\end{claim}

We start from the set of milestones described above, that is for the exponential bucketing with exponent 2. That is we start with $1, 2, 4, 8, 16, ...,2^{\log n-1}, 2^{\log n}$, for $k=0$.
For $k> 0$, we will add milestones, and for $k< 0$ we will remove milestones.

Let us start with $k\in [0, \log n -2]$.
The construction is iterative. 
From $k$ to $k+1$, we take every interval between two milestones that is non empty, and add the number in the middle. 
For example, for $k=1$, the set of milestones is: $1, 2, 3, 4, 6, 8, 12, 16, ...,  2^{\log n-1}, (3/2)2^{\log n-1}, 2^{\log n}$.
Note that the construction is well-defined as the number of elements in a non-empty interval is always odd.
At each iteration, the number of milestones is at most doubled, but as some intervals are empty, the precise increase is smaller.
It is not difficult, to compute that the precise number of milestones is: $2^k(\log n -k+1)$. 
The approximation ratio is bounded by the largest ratio between two consecutive milestones, which is $1+\frac{1}{2^k}$ (it is reached for a milestone of the original $k=0$ construction and its successor). If we push the construction one step further with $k=\log n -1$, then we get a set of milestones of size $2^{\log n -1}(\log n - \log n +1 +1)=n$, and in this case all the numbers are milestones, and the approximation is 1.

We now tackle the case $k<0$. 
Again we use an iterative construction. 
From $k$ to $k-1$, we remove one in two milestones, except for the last one. 
For example, $k=-1$ corresponds to $1,4,16,..., 2^u, 2^{\log n}$, where $u$ is the largest even integer strictly smaller than $\log n$.
At step~$k$, there are at most $\frac{\log n}{2^{-k}}+1$ milestones, and the approximation is at most $2^{2^{-k}}$. (Here we do coarser calculations, because the exact result depends more closely on the exact value of $n$, and we do not need such precision).
This finishes the proof of the claim.

Note that  the computations only deal with sums of powers of 2, thus they can all be done easily on the binary representations of the weights.

Now if $n$ is not a power of 2, we can set the set of milestones to stop after the first power of 2 that is larger than $n$. This only implies a multiplicative factor 2 for the number of milestones. Now, taking the logarithm to evaluate the space used in memory to encode the milestones, we get the result of the lemma (with the extra bit coming from the size that are not powers of 2).

\end{proof}


\section{MST construction algorithm}
\label{sec:construction}

In this section, we give a distributed algorithm for minimum spanning tree. 
This algorithm is not self-stabilizing; the self-stabilizing part will be taken care of in Section~\ref{sec:labeling-algorithm}. 
As our main goal is to use small space, we do not optimize the time of this construction algorithm (except that we want it to be polynomial), and keep it as simple as possible. 

\begin{lemma}
There exists a distributed algorithm in space $O(s + \log n)$ that builds a minimum spanning tree in polynomial time.
\end{lemma} 

\begin{proof}
Our algorithm is a distributed version of Kruskal's algorithm.
Remember that Kruskal's algorithm sorts the edges by increasing weight, and then adds the edges to the tree one by one, discarding any edge that would close a cycle. 
There are several modifications to perform in order to make this algorithm work in our framework.
 
First, to consider the edges by increasing weights, we work in phases, each phase corresponding to a specific value that a weight can take. 
We have a phase for the possible edges of weight 1, then a phase for those of weight 2, etc. 

Then in order to have these phases somehow synchronized on the whole network, and to avoid simultaneous additions of edges of the same weight (which could create cycles), we use a token. 
This token visits the nodes of the graph, and only the node with the token will be allowed to add edges to the tree. 
The token transports the information about the weight of the current phase. 
To make the token visit all the nodes we need to build a spanning tree beforehand, and we make the token traverse the tree. 
We use the same spanning tree construction and token circulation as in \cite{BlinBD18}, inspired by the tree algorithm of \cite{DattaLV11} and the token algorithm of \cite{PetitV99}. These algorithms ensure, under the distributed unfair scheduler, stabilization to a proper token circulation in $O(n)$ rounds, and they use only $O(\log n)$ bits of memory.

Finally, a node must be able to decide locally whether adding a specific edge would close a cycle or not. 
To do so every node will hold a \emph{component name}. This name is the minimum  identifier in the connected component of the node, in the current state of the tree. 
This way, a node can safely add an edge if its component name is different from the component name of the other endpoint. 
To maintain this name, we will need to perform a traversal of a part of the tree every time we add an edge. This is also known to be doable in our setting in polynomial time and logarithmic memory~\cite{BuiDPV07}.

The space complexity is in $O(s+\log n)$, because we just need to store a constant number of IDs, weights and additional $O(\log n)$-size objects. 
The time complexity in terms of rounds is polynomial. Indeed the number of phases is polynomial, because we consider polynomial weights, and each phase lasts a polynomial number of rounds (because the primitives we use, token circulation and traversal of the tree, are known to be polynomial~\cite{BlinBD18, DattaLV11, PetitV99, BuiDPV07}). 
\end{proof}
 
\paragraph*{Augmented MST}
Actually, for the next steps (described in Section~\ref{sec:labeling-algorithm}), we will need a bit more than the minimum spanning tree. 
We want to have for each node: an orientation to an arbitrary root, the number of nodes that are descendants of this node (including the node itself), and the total number of nodes in the graph. 
These are easy to compute by simple traversals.


\section{Distributed proof of MST}
\label{sec:distributed-proof}

The second part of our algorithm, that makes it self-stabilizing, consists in labeling the nodes with a distributed proof. 
This labeling certifies the correctness of the minimum spanning tree.
It comes from a proof-labeling scheme described in \cite{KormanK07}. It is necessary here to describe the scheme of \cite{KormanK07} with some details (even though it is not our work) in order to allow the reader to understand the next section without reading the full version in \cite{KormanK07}. 
For an intuitive but more in depth description of the main ideas behind the scheme, we refer to~\cite{Feuilloley19}.

\begin{lemma}[\cite{KormanK07}]
There exists a distributed proof of size $O(\log n \cdot s)$ for minimum spanning trees.
\end{lemma}

\begin{proof}[Sketch of the scheme and of the proof]
The labeling is actually in two parts. 
The first part certifies the acyclicity of the tree. 
It is well known that acyclicity can be certified with $O(\log n)$ bits \cite{KormanKP10}, for example with each node storing its number of descendants, thus we focus on the second part. 

The second part of the scheme has a recursive shape. Let us first describe the top-most level of it. 
The prover chooses a node to be the \emph{center} of the tree, and orients the tree towards this node such that it becomes a root. 
This node has some $x$ subtrees, that is, if we remove this center from the graph, there would be $x$ trees. 
The prover gives a distinct number to each subtree from 1 to $x$. 
Every node (except the center) is labeled with the number of its subtree. Also every node is given the maximum weight on its path to the center along the edges of the tree. 
It is rather easy to check that the correctness of these pieces of information can be checked locally.

To show why this is useful, consider a node~$u$ of a subtree~$A$, that is adjacent in the graph (but not in the tree) to a node~$v$ of a subtree~$B$. 
Thanks to the subtree numbers in their labels, the nodes~$u$ and~$v$ know that they do not belong to the same subtree.
We claim that they can check whether adding $(u,v)$ to the tree could result in a smaller weight tree (which would contradict the fact that the selected edges form an MST). 
First, remember that adding $(u,v)$ would lead to a lighter tree, if and only if, the path from $u$ to $v$ (in the tree) has an edge that is heavier than $(u,v)$ (and thus that could be replaced by $(u,v)$). Therefore, $u$ and $v$ only have to look for such an edge. 
The path from $u$ to $v$ must go through the root, because these nodes are in different subtrees. 
Thus the maximum weight on this path is either the maximum weight from $u$ to the root, or from $v$ to the root. As the nodes are given these maximum weights on the paths to the root in their labels, they can check whether $(u,v)$ is lighter or not than the heaviest edge on the path. 

To allow the nodes to test for all the non-selected edges, we need to handle adjacent nodes that \emph{are} in the same subtree. 
We do so by choosing recursively a center in each subtree, and providing the same information as for the top level. 
That is, every node will have information regarding the first center, its second center (\emph{i.e.} the center of its subtree of the first center), its third center (\emph{i.e.} the center of its subtree of the second center) etc. until the level where it is itself a center.
Thanks to this recursive structure, for every pair of nodes, there is a level of the recursion for which they are assigned to two distinct subtrees (or to a subtree and its root), and the checking can be done in the same way as for the top-most level.

To be sure to use only $O(\log n \cdot s)$ space, for all this information about all the centers of a node, we need the above scheme to have two properties. 
First, we need the centers to be placed in a balanced manner. 
Precisely, we want the center of a (sub)tree $T$ to be at a node such that every subtree has size at most $|T|/2$. 
(Such a node always exists and can be computed in a simple way, as described in the next section.)
This balance for the centers implies that there is at most $O(\log n)$ centers per node.
Second, we need that for each node, the concatenation of all its subtree numbers is not too long. 
To get this, it has been proved (see~\cite{GavoilleKKPP01}) that it is enough that, for each center, the subtrees are numbered by decreasing number of nodes (the largest subtree gets number 1, the second largest gets number 2, etc.).

Also, remember that we need to have an orientation towards the center, for each of the $O(\log n)$ centers a node has. 
This takes $O(\log^2n)$ bits if we use the ID of the parent as a pointer.
Instead, we have only one orientation of the full tree encoded by IDs, and the other orientations are encoded with respect to this original orientation. 
Concretely, every node will store whether the center of the current level is in the direction of its parent in the original orientation, or if it is in the direction of one of its children. 
Surprisingly, this is enough to describe and check each orientation, and it takes only constant number of bits per level. (See \cite{KormanK07} or \cite{Feuilloley19}, to see for example why not knowing the precise child-parent relation is not problematic.)
\end{proof}


\section{Certification labeling algorithm}
\label{sec:labeling-algorithm}

In this section, we describe an algorithm that builds the labeling of Section~\ref{sec:distributed-proof}.
Remember that thanks to Section~\ref{sec:construction}, we start with a tree that has an orientation toward a root and whose nodes are labeled with the number of nodes in their subtrees. Thus the first part of the labeling of Section~\ref{sec:construction} is already present. 

\medskip

\begin{lemma}
Given a tree with an orientation to a root and subtree sizes, there exists an algorithm that builds the labeling of Section~\ref{sec:distributed-proof} in space $O(\log n \cdot s)$ and polynomial time.
\end{lemma}

\begin{proof}
We first describe the algorithm, and then highlight some key properties. 
The algorithm takes as input a tree, and certifies it as a minimum spanning tree. 
The construction follows the recursive description of the labeling of Section~\ref{sec:distributed-proof}. 

Before any computation, we do a copy of the pointers, subtree number, etc. We will modify the copies, but we need to keep the original labels.
Let us first describe the computation of a center. 
(This description is for the first center, we explain later how to adapt it to the other phases of the labeling.) 
Basically, we will move the root of the tree until it satisfies the property of a proper center. (Remember that this property is that the center separates the tree into parts that have size at most half the size of the full tree.)
Our first candidate for a center is then the root of the tree. 
Thanks to the subtree sizes that every node holds, the root can detect whether it is in a central position or not.
Indeed, as the subtree size of the root is the number of nodes in the full tree, the root can easily check whether its neighbors have subtree sizes at most half this number.
If the root is not in a central position, then we transfer the root to the neighbor having the largest subtree size. 
This transfer of the root implies several computations:

\begin{itemize}
\item the old root designates the new root
\item the old root orients its pointer to the new root
\item the new root, erase its pointer and takes the root label
\item the old root takes as new subtree size its old  subtree size, minus the old subtree size of the new root
\item the new root takes as subtree size the old subtree size of the old root.
\end{itemize} 

Thanks to this step, we are in the same position as before (that is, we have a tree with a correct orientation, and correct subtree sizes), but in addition, the root is in a more central position, in the sense that the largest subtree size is smaller than before the transfer. 
After $O(n)$ such moves, the root is in a central position.  

Once the center is validated, we have to label each node with: the orientation to the center, the maximum weight on the path to the center, and the subtree number. 
There is a difficulty here. 
Remember that only the center can decide which subtree gets which number, because these numbers depend on the relative sizes of all the subtrees.
Thus the center has to announce the subtree numbers \emph{e.g.} ``the subtree with root $ID_v$ has number 3''. 
It is not possible to announce all these numbers at once. 
Indeed, the list of these numbers can have length of order $n$, therefore we cannot write the whole list in the state of the center.
Instead the center will announce a first pair identifier-number, then wait for the node with this identifier to confirm this information, and then go to the second etc. 
Once the root of a subtree receives its subtree number, it broadcasts this information to all its subtree using a snap-stabilizing (\emph{i.e.,} self-stabilizing with a stabilization time of 0 rounds) Propagation of Information with Feedback algorithm. Bui \emph{et al.}~\cite{BuiDPV07} provide such an algorithm with constant space requirement and polynomial completion time (in rounds) on trees, that meets our requirements. 
In the same wave, the maximum weight to the center is  computed by keeping and updating the maximum weight seen so far, while descending in the tree. 
The orientation is even easier to store: just copy the current orientation. 

Once the broadcast waves have come back to the root of the subtree, every node has, in its label, all the information it needs to store for this phase. 
And then we can start immediately a new recursive call, looking for the next center. 
Indeed we have a tree, with a root, with a proper orientation, and with correct subtree numbers. 

Now, let us highlight some key points of the algorithm. 

\begin{itemize}
\item Once a center has finished launching the computations in each subtree, it becomes silent, and will not change its state, except if there is a reset.
\item As soon as two adjacent nodes (in the graph) become centers, they can start checking their labelings to see if the distributed proof makes sense, at least for this edge (and launch a reset if it is not the case). 
\item Once the root has launched the computation in the subtrees, these subtrees will run independent computations (except for the cycle checking mentioned in previous item). 
This means that the scheduler can delay the computation in one subtree, by making a series of recursive calls in the other subtrees, but at some point these recursive calls will end up by having all nodes as centers, thus disabled, and the scheduler will have to enable the remaining nodes.
\item Suppose that we are in a subtree, looking for the center for the second recursive call. 
The main center is already  chosen, and thanks to the pieces of information stored, the nodes can check that  it exists. 
But they cannot check whether it was placed in a central position, because we are reusing the subtree size variables.
Thus if we start from an initial configuration that corresponds to this situation, we are not following the correct construction described in Section~\ref{sec:distributed-proof}. 
This means that we could end up with a larger memory than what we claimed.
What we do is that we control the size used. 
If we end up using more than $\alpha \cdot \log n \cdot s$ bits, for a constant $\alpha$ large enough to allow correct computations, then we launch a reset.
Note that it may actually be the case that the centers are not perfectly placed, but that their positions are good enough to ensure that the memory size does not cross the $\alpha \cdot \log n \cdot s$ limit; in this case we do not detect it, and the outcome is correct. 
\end{itemize}

When the computations ends on all nodes, every node is labeled as specified in Section~\ref{sec:distributed-proof}, and the algorithm becomes silent.  
\end{proof}


\section{Conclusion}

This paper presents the first self-stabilizing algorithm using approximation to reduce memory usage. A step in this construction is the design of a polynomial-time silent self-stabilizing algorithm for MST construction using $O(\log n \cdot s)$ bits of memory. This later algorithm uses an unusual two-phase approach: building and then certifying the solution. We believe that this modular approach is the key to go down to complexity $O(\log n \cdot s)$, and we think that it is an interesting problem to formally prove this intuition.\footnote{Some elements indicating that at least the classic techniques cannot avoid a modular approach can be be found in \cite{Feuilloley20}}

\DeclareUrlCommand{\Doi}{\urlstyle{same}}
\renewcommand{\doi}[1]{\href{https://doi.org/#1}{\footnotesize\sf doi:\Doi{#1}}}

\bibliographystyle{plainnat}
\bibliography{biblio-approx-MST}

\end{document}